\documentclass[twoside]{aiml22}

\usepackage{aiml22macro}

\usepackage{graphicx}
\usepackage{amsmath}
\usepackage{amssymb}
\usepackage{amsfonts}
\usepackage{xspace}
\usepackage{multirow}
\usepackage{comment}
\usepackage{enumitem}
\usepackage{tikz}
\usepackage{tikz-cd}
\usepackage{subcaption}
\usepackage{graphicx}
\usepackage{wrapfig}
\usepackage{hhline}
\usepackage{array}

\usetikzlibrary{shapes.geometric}

\newtheorem{observation}[thm]{Observation}

\newcommand{\todo}[1]{\textbf{\color{red}#1}}

\newcommand{\defword}{Def.}
\newcommand{\theoremword}{Th.}

\newcommand{\figureword}{Fig.}
\newcommand{\lemmaword}{Lem.}

\newcommand{\sectionword}{Sec.}
\newcommand{\appendixword}{Appendix}
\newcommand{\remarkword}{Remark}

\newcommand{\restr}[2]{\ensuremath{\left.#1\right|_{#2}}}
\newcommand{\tuple}[1]{\langle #1 \rangle}
\newcommand{\powerset}[1]{2^{#1}}
\newcommand{\emptyseq}{[\;]}
\newcommand{\seq}[1]{[#1]}
\newcommand{\apseq}[2]{#1 :: #2}
\renewcommand{\O}{\mathcal{O}}
\newcommand{\yes}{\checkmark}

\newcommand{\E}{\textbf{E}\xspace}
\newcommand{\F}{\textbf{F}\xspace}
\newcommand{\FCM}{\textbf{F\texttt{+}(CM)}\xspace}
\newcommand{\G}{\textbf{G}\xspace}
\newcommand{\GCEM}{\textbf{G\texttt{+}(CEM)}}
\newcommand{\KLMP}{\textbf{P}\xspace}
\newcommand{\PCL}{\textbf{PCL}\xspace}
\newcommand{\PCA}{\textbf{PCA}\xspace}
\newcommand{\PCAC}{\textbf{VTA}\xspace}
\newcommand{\VTA}{\textbf{VTA}\xspace}
\newcommand{\Sfive}{\textbf{S5}\xspace}

\newcommand{\Ob}[2]{\bigcirc(#1\!\mid\!#2)}
\newcommand{\PropVars}{\textit{Var}}
\newcommand{\Language}{\mathcal{F}}
\newcommand{\Models}{\mathcal{M}}

\newcommand{\logicsymb}{\mathcal{L}}

\newcommand{\propencoding}[1]{\mathcal{P}(#1)}

\newcommand{\SubF}[1]{\textit{Sub$\mathcal{F}$}(#1)}
\newcommand{\SatForm}[2]{\textit{Sat$\mathcal{F}$}_{#1}(#2)}
\newcommand{\Cond}[1]{\text{Cond}(#1)}
\newcommand{\CondSet}{\mathcal{A}}
\newcommand{\worlds}{W}
\newcommand{\Val}{\mathbb{V}}
\newcommand{\preff}{\succ}
\newcommand{\preffeq}{\succeq}
\newcommand{\preffeqc}{\succeq_c}
\newcommand{\npreffeq}{\not\succeq}
\newcommand{\preffequiv}{\approx}
\newcommand{\preffsort}{\preffeq_S}
\newcommand{\Bet}[2]{\textit{Bet}_{#1}(#2)}
\newcommand{\maximalSet}[2]{\max(#2)}
\newcommand{\maximalSetM}[1]{\maximalSet{\preff}{#1}}
\newcommand{\tset}[2]{||#1||^{#2}}
\newcommand{\best}[2]{best_{#1}(#2)}
\newcommand{\labels}{L}
\newcommand{\labeling}{\mathcal{B}}
\newcommand{\flatten}{\textit{flat}}
\newcommand{\generated}[1]{\textit{gen}(#1)}
\newcommand{\prot}{prot}
\newcommand{\rep}{rep}

\newcommand{\Orbit}[3]{\textit{Orbit}(#1, #2, #3)}
\newcommand{\AllOrbit}[2]{\textit{Cloud}(#1, #2)}
\newcommand{\MaxChain}[2]{\textit{Ray}(#1, #2)}
\newcommand{\MaxOrbChain}[3]{\textit{RayOrb}(#1, #2, #3)}

\newcommand{\worldsb}{U}
\newcommand{\preffeqb}{\preffeq_U}
\newcommand{\preffb}{\preff_U}
\newcommand{\preffeql}{\preffeq
_L}
\newcommand{\preffl}{\preff_L}
\newcommand{\succl}{\prec_L}
\newcommand{\ccsymbol}{\mathfrak{C}}

\newcommand{\preffeqcc}{\preffeq^{gen}}

\newcommand{\worldscc}{\worlds^{gen}}
\newcommand{\Valcc}{\Val^{gen}}
\newcommand{\preffeqsub}{\preffeq^N}
\newcommand{\preffsub}{\preff^N}

\newcommand{\Valsub}{\Val^N}

\newcommand{\PosBox}{\textit{Box}^+(\varphi, M)}
\newcommand{\NegBox}{\textit{Box}^-(\varphi, M)}
\newcommand{\PosOb}{\textit{Ob}^+(\varphi, M)}
\newcommand{\NegOb}{\textit{Ob}^-(\varphi, M)}
\newcommand{\Fal}{\textit{Fal}(\varphi, M)}
\newcommand{\FalBox}{\textit{Fal}^{\,\Box}(\varphi, M)}
\newcommand{\FalOb}{\textit{Fal}^{\,\bigcirc}(\varphi, M)}

\newcommand{\fallabel}[1]{\texttt{fal}_{#1}}
\newcommand{\orblabel}[1]{\texttt{orb}_{#1}}
\newcommand{\allorblabel}[1]{\texttt{cloud}_{#1}}
\newcommand{\maxchainlabel}{\texttt{ray}}
\newcommand{\betmaxchainlabel}[1]{\texttt{c-ray}_{#1}}
\newcommand{\grfallabel}[1]{\texttt{group}_{#1}}
\newcommand{\grbetmaxchainlabel}[1]{\texttt{g-ray}_{#1}}

\newcommand{\condi}{\textup{(\textit{i})}\xspace}
\newcommand{\condii}{\textup{(\textit{ii})}\xspace}
\newcommand{\condiii}{\textup{(\textit{iii})}\xspace}
\newcommand{\condiv}{\textup{(\textit{iv})}\xspace}
\newcommand{\iiisuitable}{\condiii-suitable\xspace}

\newcommand{\antichain}[2]{\textit{antichain}(#2)}
\newcommand{\chain}[2]{\textit{chain}(#2)}
\newcommand{\clique}[2]{\textit{clique}(#2)}
\newcommand{\selection}[3]{\textit{Sel}_{#1}(#2, #3)}
\newcommand{\MBet}[2]{\textit{MBet}_{#1}(#2)}
\newcommand{\maxseq}[3]{\textit{MaxSeq}_{#1}(#2, #3)}\newcommand{\maxseqword}{\textit{MaxSeq}}
\newcommand{\disjset}[3]{\mathcal{D}(#2, #3)}
\newcommand{\disjbest}[3]{d(#2, #3)}
\newcommand{\smc}[1]{\textit{SMC}^{\,#1}(\varphi, M)}
\newcommand{\modelclass}{{\mathfrak{M}}}

\newcommand{\falform}{\texttt{fal\_form}}
\newcommand{\falbox}{\texttt{fal\_box}}
\newcommand{\falboxset}[2]{\textit{FalBox}(#1, #2)}

\newcommand{\encodedf}[2]{F_{#1}(#2)}
\newcommand{\clausesmod}{C^{ev}(\varphi)}
\newcommand{\clausestr}{C^{trans}(\varphi)}
\newcommand{\clausestot}{C^{total}(\varphi)}
\newcommand{\clausesacyc}{C^{acyclic}(\varphi)}

\newcolumntype{P}[1]{>{\centering\arraybackslash}p{#1}}

\def\lastname{Rozplokhas}

\begin{document}

\begin{frontmatter}
  \title{LEGO-like Small Model Constructions for
 \AA qvist’s Logics}
  \author{Dmitry Rozplokhas}\footnote{dmitry@logic.at}
  \address{TU Wien, Austria}

  \begin{abstract}
  \AA qvist's logics (\E, \F, \FCM, and \G) are among the best-known systems in the long tradition of preference-based approaches for modeling conditional obligation. While the general semantics of preference models align well with philosophical intuitions, more constructive characterizations are needed to assess computational complexity and facilitate automated deduction. Existing small model constructions from conditional logics (due to Friedman and Halpern) are applicable only to \FCM and \G, while recently developed proof-theoretic characterizations leave unresolved the exact complexity of theoremhood in logic \F. In this paper, we introduce alternative small model constructions assembled from elementary building blocks, applicable uniformly to all four \AA qvist's logics. Our constructions propose alternative semantical characterizations and imply co-NP-completeness of theoremhood. Furthermore, they can be naturally encoded in classical propositional logic for automated deduction.
  \end{abstract}

 \end{frontmatter}

\section{Introduction}

In deontic logic, the
analysis of various normative scenarios and deontic paradoxes led to the formalization of obligations as conditionals, i.e. as dyadic modalities $\Ob{\gamma}{\alpha}$ read as ``$\gamma$ is obligatory if $\alpha$ holds''.
Traditionally these modalities are formalized using preference-based logics, inspired by the rational choice theory and introduced in the deontic context by Hansson~\cite{Hansson_deontic_logic}.
This approach considers preference models --- a kind of relational models with a ``relative goodness'' relation between worlds; a conditional obligation $\Ob{\gamma}{\alpha}$ is satisfied when $\gamma$ is true in the ``best'' worlds satisfying $\alpha$.
\AA qvist~\cite{Aqvist} formalized these ideas using the language of modal logic and his framework now serves as one of the standard implementations of the preference-based approach in deontic logic.
Initially, the framework comprised three logics of increasing deductive strength: the logic \E that places no restrictions on the preference relation, the logic \F that considers \emph{limited} preference relations to rule out contradictory obligations, and the logic \G that assumes the preference relation to be a total and limited preorder.
A later addition to \AA qvist's family is the logic \FCM~\cite{max_vs_opt} that drops the totality assumption and considers \emph{smooth} preorders, and it is axiomatized by extending \F with the cautious monotony principle, well-known in non-monotonic reasoning~\cite{Gabbay_nonmonotonic_reasoning}. 
These four logics provide a useful scale against which various deontic scenarios can be evaluated.

At the same time, preference models were applied in the neighboring field of conditional reasoning.
Notable examples of conditional frameworks defined in terms of preference models include Lewis's family of logic for counterfactuals~\cite{Lewis_counterfactuals}, Burgess's preferential conditional logic \PCL and its extensions~\cite{Burgess}, and KLM logic for non-monotonic reasoning~\cite{KLM}.
These frameworks consider similar models, so there is an intersection with the \AA qvist family: \FCM coincides with Burgess's logic \PCA, and \G coincides with Lewis's logic \VTA (using the terminology of~\cite{Mariana_thesis}), while the flat fragment of both these logics coincides with KLM logic \KLMP.
However, it is not common in conditional logics 
to consider relaxed notions of preference relations,
since assuming transitivity and smoothness is necessary for a well-behaved consequence relation in the logic~\cite{non_smooth_preferences}.
On the other hand, when the preference relation is treated as comparative goodness, the adequacy of these assumptions becomes controversial (see, e.g., \cite{non_smooth_preferences} and \cite[Sec. 2.3]{SEP_preferences} for an overview of counterexamples).
For this reason, the weaker logics \E and \F play an important role
in the normative reasoning context.

In recent decades, significant progress has been made in exploring variations of preference-model characterization for \AA qvist's logics and their corresponding axiomatizations, surveyed in~\cite{Xavier_handbook_survey}. Now, there's a growing focus on the computational properties of these logics, which is the main motivation for this paper too.
In~\cite{Xavier_handbook_survey} the decidability of theoremhood for all four logics is proven through alternative semantics based on selection functions, and embedding of the weakest logic \E into Higher-Order Logic (HOL) from~\cite{E_in_HOL} is suggested as a potential approach for automated deduction. These approaches however are not suitable for assessment of the exact complexity of logics, which requires more constructive characterizations.
One such characterization came from the proof-theoretic side in the form of
cut-free hypersequent calculi, developed recently for all four \AA qvist's logics~\cite{E_hypersequents,F_hypersequents,FCM_G_hypersequents}. 
For \E, \FCM and \G the proof search in the calculi has optimal co-NP complexity, and polynomial-size preference countermodels can be reconstructed from failed derivations~\cite{E_hypersequents,Mariana_thesis}. 
 At the same time, the limitedness condition of \F seems 
difficult to handle both model-theoretically and proof-theoretically.
The calculus for \F~\cite{F_hypersequents}, which is an even more complicated variation of the calculi for logic {\bf GL}~\cite{Valentini_GL},
gives only a co-NEXP upper bound for theoremhood (which is the best estimation so far) and no countermodel construction. 


Another powerful approach for establishing computational complexity of conditional logics is \emph{small model constructions} proposed by Friedman and Halpern~\cite{Friedman_Halpern_small_model_constructions} for Burgess' logic \PCL and its extensions, which transforms any satisfying model into a satisfying model of bounded size. Their approach covers in particular extensions \PCA and \PCAC (i.e. \AA qvist's logics \FCM and \G), and establishes co-NP-completeness of theoremhood for them.  However, this approach significantly relies on the smoothness and transitivity of the preference relation and therefore is not applicable for weaker logics \E and \F (see \remarkword~\ref{rem:Friedman_Halpern_FCM} for details).

In this paper, we propose alternative small model constructions
to handle all four \AA qvist's logics uniformly.
We compose a model of polynomial size by assembling elementary building blocks (chains, antichains, and cliques of worlds selected from any given model) like LEGO.
We provide sufficient conditions for such construction to be a countermodel 
and define a suitable construction for each \AA qvist's logic.
There are two main applications for our constructions, obtained uniformly for all logics.

{\bf Alternative semantical characterizations of theoremhood.}
Our results imply that theoremhood can be characterized by finite models. Moreover, for finite models the complicated properties of limitedness and smoothness (which are not frame properties) can be replaced by natural frame properties: acyclicity and transitivity of the preference relation, respectively.

{\bf Complexity and automated deduction.} The polynomial size of models together with easily checkable frame properties immediately imply co-NP-completeness of theoremhood (including logic \F, for which it was an open problem) and allow for natural encodings in classical propositional logics, which can be utilized for efficient automated deduction using SAT-solvers.

\section{Preliminaries}

The syntax of \AA qvist's logics extends the usual propositional language with two modalities: unary $\Box$ for necessity and binary $\Ob{\cdot}{\cdot}$ for conditional obligation. We define the formulas over the set $\PropVars$ of propositional variables.

\[ \Language ::= \; x \in \PropVars \;\mid\; \neg \Language \;\mid\; \Language \wedge \Language \;\mid\; \Box \Language \;\mid\; \Ob{\Language}{\Language}   \]

We will use small Greek letters to denote formulas. $|\varphi|$ will denote size of the formula (number of symbols), $\SubF{\varphi}$ will denote the set of all subformulas of $\varphi$ (including $\varphi$), and $\Cond{\varphi} = \{ \alpha \mid {\Ob{\gamma}{\alpha} \in \SubF{\varphi}} \}$.


\begin{definition}
A \emph{preference model} is a triple ${\tuple{W, \succeq, \Val}}$ where $W$ is a (non-empty) set of worlds, $\succeq$ is a binary relation on $W$, and $\Val \colon \PropVars \to \powerset{W}$ is a valuation function. We denote by $\worlds(M)$ the set of worlds of a given model. 
\end{definition}

The semantics of obligation is based on the notion of ``best'' worlds in the preference model. There are different definitions of bestness appearing in the literature (see \cite{general_preference_relations,max_vs_opt} for the comparison of different definitions), we will use the most common one~--- maximality: a world is a best world when there are no worlds that are \emph{strictly} more preferable. As usual we denote by $\preff$ a strict version of $\preffeq$ ($w_1 \preff w_2$ when $w_1 \preffeq w_2$ and $w_2 \npreffeq w_1$). We will use the notation $\Bet{\preff}{v} = \{ w \in W \mid w \preff v \}$ for a set of worlds strictly preferable to (better than) a given one.

\begin{definition}
For a preference model $M = {\tuple{W, \succeq, \Val}}$ and $U \subseteq W$ we define ${\maximalSetM{U} = \{ v \in U \;\mid\; \nexists u \in U \colon u \preff v \}}$. 
\end{definition}

Satisfaction of $\Ob{\gamma}{\alpha}$ is defined using this notion of bestness: $\Ob{\gamma}{\alpha}$ is true when $\gamma$ is true in all maximal worlds satisfying $\alpha$ (we will call such worlds \emph{$\alpha$-maximal}). And $\Box \beta$ is true when $\beta$ is true in all worlds (so we treat $\Box$ as the universal \Sfive modality). 

\begin{definition}{\bf(Satisfaction)}
For a preference model $M = {\tuple{W, \succeq, \Val}}$ the \emph{truth set} $\tset{\varphi}{M}$ of a formula $\varphi$ is defined inductively:
\vspace{1mm}
\begin{itemize}
\item $w \in \tset{x}{M}$ for $x \in \PropVars$ when $w \in \Val(x)$,
\item $w \in \tset{\neg \psi}{M}$ when $w \not\in \tset{\psi}{M}$,
\item $w \in \tset{\psi_1 \wedge \psi_2}{M}$ when $w \in \tset{\psi_1}{M}$ and $w \in \tset{\psi_2}{M}$,
\item $w \in \tset{\Box \beta}{M}$ when $\tset{\beta}{M} = W$,
\item $w \in \tset{\Ob{\gamma}{\alpha}}{M}$ when $\maximalSetM{\tset{\alpha}{M}} \subseteq \tset{\gamma}{M}$.
\end{itemize}
\vspace{2mm}
We say that $w$ satisfies $\varphi$ in $M$ (denoted $M, w \models \varphi$) when $w \in \tset{\varphi}{M}$, and that $M$ validates $\varphi$ (denoted $M \models \varphi$) when $\tset{\varphi}{M} = W$. For $U \subseteq W$ we denote a set of formulas satisfiable in $U$ as $\SatForm{M}{U} = \{ \psi \mid \exists u \in U : M, u \models \psi \}$.   
\end{definition}

Notice that the satisfaction of both $\Box \beta$ and $\Ob{\gamma}{\alpha}$ does not depend on the world of evaluation.

Different \AA qvist's logics are defined by different classes of preference models. Some of these classes are defined using the properties of preference relation $\preffeq$ in the model, we will use two properties: transitivity ($\preffeq$ is transitive when $w_1 \preffeq w_2$ and $w_2 \preffeq w_3$ imply $w_1 \preffeq w_3$) and totalness ($\preffeq$ is total when for any $w_1, w_2 \in W$ either $w_1 \preffeq w_2$ or $w_2 \preffeq w_1$).
Another property used for the characterization of deontic logic is what Lewis called ``limit assumption'', which ensures the existence of best worlds. The are different formal definitions of this assumption in the literature, we will use two versions from~\cite{Xavier_handbook_survey}: \emph{limitedness} and \emph{smoothness}.

\begin{definition}{\bf(Limit conditions)}
Let $M = {(W, \succeq, \Val)} \in \Models$.
$M$ is \emph{limited} when for any formula $\alpha$ if  $\tset{\alpha}{M} \neq \emptyset$ then $\maximalSetM{\tset{\alpha}{M}} \neq \emptyset$.
$M$ is \emph{smooth} when for any formula $\alpha$ and any world $w \in \tset{\alpha}{M}$ there exists $u \in \maximalSetM{\tset{\alpha}{M}}$ such that either $u = w$ or $u \preff w$.
\end{definition}

\begin{figure}
    \centering
    \begin{tabular}{|c|P{15mm}|P{15mm}|P{15mm}|P{15mm}|}
    \hline
    \multirow{2}{*}{Logic} & \multicolumn{2}{c|}{Limit conditions} & \multicolumn{2}{c|}{Properties of $\preffeq$}  \\
    \hhline{|~|-|-|-|-|}
     & limited & smooth & transitive & total  \\
    \hline
    $\E$ & & & &  \\
    \hline
    $\F$ & \yes & & &  \\
    \hline
    $\FCM$ & & \yes & \yes & \\
    \hline
    $\G$ & & \yes & \yes & \yes \\
    \hline
    \end{tabular}
    \caption{Preference-semantical characterizations for \AA qvist's logics~\cite[Tab. 1 and 2]{Xavier_handbook_survey} (with maximality as the notion of bestness). }
    \label{tab:initial_taxonomy}
\end{figure}

We rely on the semantical characterizations of the four \AA qvist logics in \figureword~\ref{tab:initial_taxonomy}, which are presented (among various other characterizations) in~\cite{Xavier_handbook_survey}.

\begin{definition}
Formula $\varphi$ is a theorem of \AA qvist's logic $\logicsymb$ iff $M \vDash \varphi$ for any preference model $M$ that satisfies model conditions for logic $\logicsymb$ in \figureword~\ref{tab:initial_taxonomy}.
\end{definition}

We will call a preference model $M$ a \emph{countermodel for a formula $\varphi$} if $M \not\models \varphi$ and we will further call it an \emph{$\logicsymb$-countermodel} if it belongs to a class of models corresponding to a logic $\logicsymb$ from \figureword~\ref{tab:initial_taxonomy}.

\section{Small Model Constructions}

This section contains the main technical result of the paper: for every logic $\logicsymb$ from the \AA qvist family, we will show how an arbitrary $\logicsymb$-countermodel $M$ for a formula $\varphi$ can be transformed into an $\logicsymb$-countermodel with the number of worlds bounded polynomially w.r.t. $|\varphi|$.
We will achieve this by selecting a finite number of worlds from $M$, adding copies for some of them, and defining a new preference relation on the selected worlds without changing the valuation. We call such transformation a \emph{rearrangement} of a model.

\begin{definition}
We say that a model $M' = \tuple{W', \preffeq', \Val'}$ \emph{rearranges} the model $M = \tuple{W, \preffeq, \Val}$ when there exists a \emph{prototype function} $\prot \colon W' \to W$ such that $w' \in \Val'(x)$ is equivalent to $\prot(w) \in \Val(x)$ for all $x \in \PropVars$.
\end{definition}



Our main goal for the rearranged model is to have each of its worlds satisfying the same subformulas of $\varphi$ as its prototype does. 
Evaluation of a formula in a world involves other worlds only in the cases of $\Box$ and $\Ob{\cdot}{\cdot}$ modalities. Therefore, we only need to ensure that the rearranged model validates the same modalities among subformulas of $\varphi$ as the original model does, while the satisfaction (and non-satisfaction) of other subformulas will be preserved in the rearranged model automatically. 

We will examine the cases of validated and non-validated modalities separately. Let us denote by $\PosBox$ (resp. $\PosOb$) the set of subformulas of $\varphi$ of the form $\Box \beta$ (resp. $\Ob{\gamma}{\alpha}$) that are validated by $M$, and by $\NegBox$ and $\NegOb$ the sets of subformulas of $\varphi$ of the corresponding form that are not validated by $M$. 
To falsify $\Box \beta \in \NegBox$ and $\Ob{\gamma}{\alpha} \in \NegOb$ we need to take in $M'$ some worlds that were falsifying these modalities in $M$. 
While the evaluation of $\Box \beta$ modalities relies only on the presence of the worlds satisfying $\beta$ in the model, special care is needed to ensure that the evaluation of $\Ob{\gamma}{\alpha}$ is the same. Namely, if a world $w$ was made not $\alpha$-maximal in $M$ by some world $u \in \tset{\alpha}{M}$ such that $u \preff w$ we need to preserve this violation of maximality in $M'$. Conversely, we need to ensure that we are not violating $\alpha$-maximality in $M'$ for the world falsifying $\Ob{\gamma}{\alpha}$.  
This reasoning leads to the following four conditions sufficient to ensure that a rearranged model $M'$ is a countermodel for $\varphi$.

\begin{theorem}
\label{lem:rearranged_countermodel_conditions}
Suppose $M = \tuple{W, \preffeq, \Val}$ is a countermodel for $\varphi$ and ${M' = \tuple{W', \preffeq', \Val'}}$ rearranges $M$ with the prototype function $\prot \colon W' \to W$. Then the following conditions are sufficient for $M' \not\models \varphi$.
\begin{enumerate}
\item[\condi] There exists $v' \in W'$ such that $M, \prot(v') \not\models \varphi$.
\item[\condii] For any $\Box \beta \in \NegBox$ there exists $v' \in W'$ such that $M, \prot(v') \not\models \beta$.
\item[\condiii] For any ${\Ob{\gamma}{\alpha} \in \NegOb}$ there exists ${v' \in W'}$ such that ${\prot(v') \in \maximalSetM{\tset{\alpha}{M}} \setminus \tset{\gamma}{M}}$ and for all ${u' \preff' v}'$ holds ${\prot(u') \preff \prot(v')}$. 
\item[\condiv] For any $w' \in W'$, for all $\Ob{\gamma}{\alpha} \in \PosOb$ if there exists $u \preff \prot(w')$ such that $M, u \models \alpha$ then there exists $s' \preff' w'$ such that $M, \prot(s') \models \alpha$.
\end{enumerate}
\end{theorem}
\begin{proof}
We will prove a generalized statement: for any $w' \in W'$ and any $\psi \in \SubF{\varphi}$ holds $M', w' \models \psi$ iff $M, \prot(w') \models \psi$. Then $M' \not\models \varphi$ follows from condition \condi. The proof is by induction on $\psi$
with case analysis on $\psi$ belonging to $\PosBox$ or $\NegBox$ for $\psi = \Box \beta$ and on $\psi$ belonging to $\PosOb$ or $\NegOb$ for $\psi = \Ob{\gamma}{\alpha}$. Conditions \condii, \condiii, and \condiv directly cover cases $\psi \in \NegBox$, $\psi \in \NegOb$, and $\psi \in \PosOb$ respectively (see appendix~\ref{sec:inductive_details} for details).
\end{proof}

Ensuring conditions \condi and \condii is simple: we need to take arbitrary worlds from $(W \setminus \tset{\varphi}{M})$ and from $(W \setminus \tset{\beta}{M})$ for each $\Box \beta \in \NegBox$. For this, we will use a \emph{representative function} $\rep \colon (\powerset{W} \setminus \{ \emptyset \}) \to W$ that chooses an element $\rep(S) \in S$ from any given non-empty subset $S$ of $W$ (thus, we use the axiom of choice explicitly in our construction). We will also need representatives of ${(\maximalSetM{\tset{\alpha}{M}} \setminus \tset{\gamma}{M})}$ for every $\Ob{\gamma}{\alpha} \in \NegOb$ for condition \condiii. Let us denote the set of all such falsifying worlds $\Fal$.

\begin{definition}{\bf(Falsifying worlds)}
For a model $M = \tuple{W, \preffeq, \Val}$ such that $M \not\models \varphi$, $\Fal = {\rep(W \setminus \tset{\varphi}{M}) \cup \FalBox \cup \FalOb}$, where $\FalBox = {\{ \rep(W \setminus \tset{\beta}{M}) \mid \Box \beta \in \NegBox \}}$, $\FalOb = {\{ \rep(\maximalSetM{\tset{\alpha}{M}} \setminus \tset{\gamma}{M}) \mid \Ob{\gamma}{\alpha} \in \NegOb \}}$.
\end{definition}

The rest of the rearranged model will be chosen to ensure the satisfaction of conditions \condiii and \condiv. 
We will represent our small model constructions as \emph{composite models}, assembled from \emph{blocks}. A \emph{block} $B$ is a finite selection of worlds from $M$ with some new preference relation on them (in our cases it will be either an empty relation, a strict linear order, or a universal relation).

\begin{definition}{\bf (Block)}
A \emph{block} on $M$ is a tuple $\tuple{\worldsb, \preffeqb}$ where $\worldsb \subseteq \worlds(M)$ and $\preffeqb$ is a binary relation on $\worldsb$.
We will use $\worlds(B)$ to refer to the set of worlds in $B$. For a given $M$ and $U \subseteq \worlds(M)$ we will consider the blocks of the following forms:
\begin{itemize}
\item $\antichain{M}{U} = \tuple{U, \preffeq^a}$, where $\preffeq^a$ is an empty relation;
\item $\chain{M}{S} = \tuple{ \{w_i\}_{i=1}^n, \preffeq^{ch} }$ if $S = [w_1, \dots, w_n]$ is a finite ordered sequence of worlds and $w_i \preffeq^{ch} w_j$ iff $i \le j$;
\item $\clique{M}{U} = \tuple{U, \preffeq^{cl}}$ where ${u_1 \preffeq_{cl} u_2}$ for all $u_1, u_2 \in U$.
\end{itemize}
\end{definition}

A composite construction consists of the number of blocks with an additional preference relation on them. Each composite construction generates a model rearranging $M$, in which the new preference relation is given by combining the relation between blocks and the relations inside blocks. To allow multiple occurrences of the same block in the construction we define the composite construction using labels and a labeling function.

\begin{definition}{\bf(Composite construction)}
A \emph{composite construction on $M$} is a tuple $\tuple{\labels, \preffeql, \labeling}$ where $\labels$ is a set of labels, $\preffeql$ is a binary relation on $\labels$, and $\labeling$ is a labeling function that maps every label from $\labels$ into a block on $M$. 
Each composite construction $\ccsymbol = \tuple{\labels, \preffeql, \labeling}$ on $M = \tuple{\worlds, \preffeq, \Val}$ \emph{generates} a model $\generated{\ccsymbol} = \tuple{\worldscc, \preffeqcc, \Valcc}$, where 
\begin{itemize}
\item $\worldscc = \{ (l, w) \mid l \in \labels, w \in \worlds(\labeling(l)) \}$;
\item ${(l_1, w_1) \preffeqcc (l_2, w_2)}$ iff either $l_1 \preffeql l_2$ or both $l_1 = l_2$ and $w_1 \preffeqb w_2$ for $\labeling(l_1) = \tuple{\worldsb, \preffeqb}$;
\item $(l, w) \in \Valcc(x)$ iff $w \in \Val(x)$.
\end{itemize}
\end{definition}


We can now simplify the conditions of \theoremword~\ref{lem:rearranged_countermodel_conditions} for models generated by composite constructions. For conditions \condi-\condiii it is enough to have every world $v$ from $\Fal$ in some block in the construction, such that this block does not have any $\preffb$-preferable worlds inside and all blocks $\preffl$-prefferable to it have only worlds from $\Bet{\preffeq}{v}$. Also, we can ensure \condiv separately for each block by either ensuring it inside this block or having another block \mbox{$\preffl$-preferred} to it that has all worlds required in \condiv. We express these condition using the following notions of block compatibility.

\begin{definition}{\bf(Block compatibility properties).} For a given model $M$ and a formula $\varphi$ we define the following properties of blocks on $M$:
\begin{itemize}
\item $\tuple{\worldsb, \preffeqb}$ is \emph{flat} when $w_1 \not\preffb w_2$ for any $w_1, w_2 \in \worldsb$.
\item $B'$ is \emph{{\condiii-suitable} for} $B$ when $\worlds(B') \subseteq \Bet{\preff}{w}$ for any $w \in \worlds(B)$.
\item $\tuple{\worldsb, \preffeqb}$ is \emph{{\condiv-safe}} when for every $w \in \worldsb$ and every ${\alpha \in \SatForm{M}{\Bet{\preff}{w}} \cap \Cond{\varphi}}$ there is $w' \preffb w$ such that  $M, w' \models \alpha$.
\item $B'$ \emph{{\condiv-covers}} $B$ when $\SatForm{M}{\Bet{\preff}{w}} \cap \Cond{\varphi} \subseteq \SatForm{M}{\worlds(B')}$ for any $w \in \worlds(B)$.
\end{itemize}    
\end{definition}

\newcommand{\rulesep}{\unskip\hfill{\vrule}\hfill\ignorespaces}

\def\circlesize{4pt}
\def\singletonsize{15pt}
\newcommand{\worldnode}[2]{
  \draw (#1,#2) node[circle,draw,inner sep=1pt,fill=gray!50,minimum size=\circlesize,semithick]{};
}

\def\dotslen{0.115cm}
\newcommand{\dotsnode}[2]{
  \draw [dotted,thick](#1-\dotslen,#2) -- (#1+\dotslen,#2);
}

\def\tipstep{0.120cm}
\def\dotshei{0.115cm}
\newcommand{\vdotsnode}[2]{
  \draw [dotted,thick](#1,#2-\dotshei) -- (#1,#2+\dotshei);
  \draw [thick,-to](#1,#2+\dotshei+\tipstep-0.05) -- (#1,#2+\dotshei+\tipstep);
}

\def\tipvstep{0.240cm}
\newcommand{\rdotsnode}[2]{
  \dotsnode{#1}{#2};
  \draw [thick,-Stealth](#1+\dotslen+\tipvstep-0.05,#2) -- (#1+\dotslen+\tipvstep,#2);
}

\newcommand{\cldotsnode}[2]{
  \dotsnode{#1}{#2};
  \draw [thick,-to](#1+\dotslen+\tipstep-0.05,#2) -- (#1+\dotslen+\tipstep,#2);
  \draw [thick,-to](#1-\dotslen-\tipstep+0.05,#2) -- (#1-\dotslen-\tipstep,#2);
}

\def\sinwh{14pt}
\newcommand{\sincomp}[6]{
  \node[draw,rectangle,rounded corners,dashed,minimum width=\sinwh,minimum height=\sinwh,label={[xshift=#5,yshift=#6]#4}] (#3) at ({#1},{#2}) {};
  \worldnode{#1}{#2};
}

\def\worldhdist{0.45cm}
\def\hcompw{50pt}
\def\hcomph{16pt}
\newcommand{\achcomp}[6]{
  \node[draw,rectangle,rounded corners,dashed,minimum width=\hcompw,minimum height=\hcomph,label={[xshift=#5,yshift=#6]#4}] (#3) at ({#1},{#2}) {};
  \worldnode{#1-\worldhdist}{#2};
  \worldnode{#1+\worldhdist}{#2};
  \dotsnode{#1}{#2};
}

\newcommand{\clcomp}[6]{
  \node[draw,rectangle,rounded corners,dashed,minimum width=\hcompw,minimum height=\hcomph,label={[xshift=#5,yshift=#6]#4}] (#3) at ({#1},{#2}) {};
  \worldnode{#1-\worldhdist}{#2};
  \worldnode{#1+\worldhdist}{#2};
  \cldotsnode{#1}{#2};
}

\def\vcompw{20pt}
\def\vcomph{45pt}
\def\chlowshift{0.1}
\newcommand{\chcomp}[6]{
  \node[draw,rectangle,rounded corners,dashed,minimum width=\vcompw,minimum height=\vcomph,label={[xshift=#5,yshift=#6]#4}] (#3) at ({#1},{#2+\chlowshift}) {};
  \worldnode{#1}{#2-\worldhdist+2*\chlowshift};
  \worldnode{#1}{#2+\worldhdist};
  \vdotsnode{#1}{#2};
}

\def\comphdist{2.3}
\def\compvdist{2}
\def\compvdistch{2.7}
\def\arroww{0.65pt}

\begin{figure}[!t]

\centering
\begin{tikzpicture}

\usetikzlibrary{shapes.geometric}

\def\framewidth{10cm}
\def\frameheight{12cm}
\draw[step=5.95cm,black] (0,0) grid (11.9,11.9);

\def\leftcenter{3cm}
\def\rightcenter{9cm}
\def\upcenter{9cm}
\def\downcenter{3cm}
\def\captionshift{70pt}
\def\falstep{35pt}


\node at (\leftcenter, \upcenter-\captionshift) {(a) $\smc{\E}$};

\def\efallev{-45pt}
\def\eorblev{-5pt}
\def\ecloudlevb{30pt}
\def\ecloudlevt{60pt}

\def\arachup{0.285cm}
\def\archhor{0.34cm}
\def\archver{0.7cm}
\def\arachhor{0.82cm}
\def\arachver{0.25cm}

\sincomp{\leftcenter-\falstep}{\upcenter+\efallev}{ef1}{$\fallabel{v_1}$}{-18}{-5};
\dotsnode{\leftcenter}{\upcenter+\efallev};
\sincomp{\leftcenter+\falstep}{\upcenter+\efallev}{efn}{$\fallabel{v_n}$}{20}{-5};

\achcomp{\leftcenter-\falstep}{\upcenter+\eorblev}{eorb1}{$\orblabel{v_1}$}{-18}{-1.5};
\dotsnode{\leftcenter}{\upcenter+\eorblev};
\achcomp{\leftcenter+\falstep}{\upcenter+\eorblev}{eorbn}{$\orblabel{v_n}$}{18}{-1.5};

\achcomp{\leftcenter}{\upcenter+\ecloudlevb}{aorb1}{$\allorblabel{1}$}{0}{-1};
\achcomp{\leftcenter+\falstep}{\upcenter+\ecloudlevt}{aorb2}{$\allorblabel{2}$}{0}{-1};
\achcomp{\leftcenter-\falstep}{\upcenter+\ecloudlevt}{aorb3}{$\allorblabel{3}$}{0}{-1};

\draw [line width=\arroww,-Stealth] (ef1) -- (eorb1);
\draw [line width=\arroww,-Stealth] (efn) -- (eorbn);
\draw [line width=\arroww,-Stealth] (aorb2) -- (aorb3);
\draw [line width=\arroww,-Stealth] (\leftcenter-\falstep,\upcenter+\eorblev+\arachup) -- (\leftcenter-\arachhor,\upcenter+\ecloudlevb-\arachver);
\draw [line width=\arroww,-Stealth] (\leftcenter+\falstep,\upcenter+\eorblev+\arachup) -- (\leftcenter+\arachhor,\upcenter+\ecloudlevb-\arachver);
\draw [line width=\arroww,-Stealth] (\leftcenter-\falstep,\upcenter+\ecloudlevt-\arachup) -- (\leftcenter-\arachhor,\upcenter+\ecloudlevb+\arachver);
\draw [line width=\arroww,-Stealth] (\leftcenter+\arachhor,\upcenter+\ecloudlevb+\arachver) -- (\leftcenter+\falstep,\upcenter+\ecloudlevt-\arachup);


\node at (\rightcenter, \upcenter-\captionshift) {(b) $\smc{\F}$};

\def\fchlev{40pt}

\sincomp{\rightcenter-\falstep}{\upcenter+\efallev}{ff1}{$\fallabel{v_1}$}{-18}{-5};
\dotsnode{\rightcenter}{\upcenter+\efallev};
\sincomp{\rightcenter+\falstep}{\upcenter+\efallev}{ffn}{$\fallabel{v_n}$}{20}{-5};

\achcomp{\rightcenter-\falstep}{\upcenter+\eorblev}{forb1}{$\orblabel{v_1}$}{-18}{-1.5};
\dotsnode{\rightcenter}{\upcenter+\eorblev};
\achcomp{\rightcenter+\falstep}{\upcenter+\eorblev}{forbn}{$\orblabel{v_n}$}{18}{-1.5};

\chcomp{\rightcenter}{\upcenter+\fchlev}{fch}{$\maxchainlabel$}{0}{-1};

\draw [line width=\arroww,-Stealth] (ff1) -- (forb1);
\draw [line width=\arroww,-Stealth] (ffn) -- (forbn);
\draw [line width=\arroww,-Stealth] (\rightcenter-\falstep,\upcenter+\eorblev+\arachup) -- (\rightcenter-\archhor,\upcenter+\fchlev-\archver);
\draw [line width=\arroww,-Stealth] (\rightcenter+\falstep,\upcenter+\eorblev+\arachup) -- (\rightcenter+\archhor,\upcenter+\fchlev-\archver);


\node at (\leftcenter, \downcenter-\captionshift) {(c) $\smc{\FCM}$};

\def\cmfallev{-30pt}
\def\cmorblev{30pt}

\sincomp{\leftcenter-\falstep}{\downcenter+\cmfallev}{cf1}{$\fallabel{v_1}$}{-18}{-5};
\dotsnode{\leftcenter}{\downcenter+\cmfallev};
\sincomp{\leftcenter+\falstep}{\downcenter+\cmfallev}{cfn}{$\fallabel{v_n}$}{20}{-5};

\chcomp{\leftcenter-\falstep}{\downcenter+\cmorblev}{corb1}{$\betmaxchainlabel{v_1}$}{5}{-3};
\dotsnode{\leftcenter}{\downcenter+\cmorblev};
\chcomp{\leftcenter+\falstep}{\downcenter+\cmorblev}{corbn}{$\betmaxchainlabel{v_n}$}{5}{-3};

\draw [line width=\arroww,-Stealth] (cf1) -- (corb1);
\draw [line width=\arroww,-Stealth] (cfn) -- (corbn);


\node at (\rightcenter, \downcenter-\captionshift) {(d) $\smc{\G}$};

\def\gdotslev{10pt}

\clcomp{\rightcenter-\falstep}{\downcenter+\cmfallev}{gf1}{$\grfallabel{S_1}$}{-23}{-1.5};
\chcomp{\rightcenter-\falstep}{\downcenter+\cmorblev}{gotb1}{$\grbetmaxchainlabel{S_1}$}{5}{-3};

\rdotsnode{\rightcenter}{\downcenter+\gdotslev};

\clcomp{\rightcenter+\falstep}{\downcenter+\cmfallev}{gfn}{$\grfallabel{S_n}$}{25}{-1.5};
\chcomp{\rightcenter+\falstep}{\downcenter+\cmorblev}{gotbn}{$\grbetmaxchainlabel{S_n}$}{5}{-3};

\draw [line width=\arroww,-Stealth] (gf1) -- (gotb1);
\draw [line width=\arroww,-Stealth] (gfn) -- (gotbn);

\end{tikzpicture}

\DeclareRobustCommand{\dotsinline}{\tikz\dotsnode{0}{0};\xspace}
\DeclareRobustCommand{\vertarrowinline}{\tikz\vdotsnode{0}{0};\xspace}
\DeclareRobustCommand{\cliquearrowinline}{\tikz\cldotsnode{0}{0};\xspace}
\DeclareRobustCommand{\arrowinline}{\tikz\rdotsnode{0}{0};\xspace}

\caption{Small model constructions for \AA qvist's logics. Gray circles represent worlds, dashed rectangles represent blocks. Symbol \raisebox{2pt}{\dotsinline} inside a block indicates an antichain, \mbox{\;\vertarrowinline\; indicates a chain}, and \;\raisebox{2pt}{\cliquearrowinline} indicates a clique. 
Solid arrows represent the preference relation $\preffeql$ between blocks: an arrow from a block $l_1$ to a block $l_2$ means $l_2 \preffeql l_1$. The arrow \,\arrowinline\, between blocks in construction $\smc{\G}$ means that there is a linear order on blocks. Note that the preference relation in constructions $\smc{\E}$ and $\smc{\F}$ is not transitive.  }

\label{fig:small_model_constructions}
\end{figure}

Using these notions we can reformulate (the weaker version of) \theoremword~\ref{lem:rearranged_countermodel_conditions} for composite constructions as follows.

\begin{theorem}
\label{lem:composite_countermodel_conditions}
Let $M \not\models \varphi$ and $\ccsymbol = \tuple{\labels, \preffeql, \labeling}$ be a composite construction on $M$. For $\generated{\ccsymbol} \not\models \varphi$ it is sufficient that:
\begin{enumerate}[label=(\alph*)]
\item For every falsifying world $v \in \Fal$ there is a label $b_v \in \labels$ such that $v \in \worlds(\labeling(b_v))$ and $\labeling(b_v)$ is flat and $\labeling(b')$ is \iiisuitable for $\labeling(b_v)$ for every $b' \preffl b_v$.
\item For every $b \in \labels$ the block $\labeling(b)$ is either \condiv-safe
or \condiv-covered
by $\labeling(b')$ for some $b' \preffl b$.
\end{enumerate}
\end{theorem}
\begin{proof}
$\generated{\ccsymbol}$ rearranges $M$ (with $\prot((l, w)) = w$), so we can apply \theoremword~\ref{lem:rearranged_countermodel_conditions}. (a) ensures conditions \condi-\condiii and (b) ensures condition \condiv.
\end{proof}


We now define composite constructions for each \AA qvist's logic satisfying the conditions from \theoremword~\ref{lem:composite_countermodel_conditions} and the model conditions for the logic from \figureword~\ref{tab:initial_taxonomy}.

\subsection{Small Model Construction for Logic \E}

In the case of logic \E there are no model conditions we need to satisfy in our countermodel, so we can use a preference relation that is non-transitive and contains cycles. In this simple case, all blocks of the countermodel construction will be antichains.

We start our composite construction with a dedicated one-world block $\antichain{M}{\{v\}}$ labeled $\fallabel{v}$ for each world $v \in \Fal$.
The simplest way to \condiv-cover such block with a \iiisuitable block containing linearly many (w.r.t. $|\varphi|$) worlds is to go through formulas from $\Cond{\varphi}$ satisfied by some world in $\Bet{\preff}{v}$ and select one representative for each. Below such selection is defined more generally, for an arbitrary set of formulas $\CondSet$ and an arbitrary set of worlds $U$ to select from.

\begin{definition}{\bf (Selection)}
For a set $U \subseteq \worlds(M)$ and a set of formulas $\CondSet$, $\selection{M}{U}{\CondSet} = \{ \rep(\tset{\alpha}{M} \cap U) \mid \alpha \in \CondSet, \tset{\alpha}{M} \cap U \neq \emptyset \}$.
\end{definition}

We can show that such a selection can \condiv-cover not only single-world blocks like $\fallabel{v}$, but any block $B$ as long as $U$ contains all worlds $\preff$-preferable to some world in $B$.

\begin{lemma}
\label{lem:cond_iv_for_selection}
If $\Bet{\preff}{w} \subseteq U$ for all $w \in \worlds(B)$ in some block $B$ then $\antichain{M}{\selection{M}{U}{\Cond{\varphi}}}$ \condiv-covers $B$.
\end{lemma}
\begin{proof}
If $\alpha \in \Cond{\varphi}$ is satisfiable in $\Bet{\preff}{w}$ for some $w \in \worlds(B)$ then there will be a representative satisfying $\alpha$ in $\selection{M}{U}{\Cond{\varphi}}$.
\end{proof}

A block to \condiv-cover $\fallabel{v}$, which we will call \emph{orbit} and label $\orblabel{v}$, can be defined as $\Orbit{M}{\varphi}{v} = \antichain{M}{\selection{M}{\Bet{\preff}{v}}{\Cond{\varphi}}}$. To \condiv-cover orbits themselves, we can make another selection, this time from the whole $\worlds(M)$, as this block does not need to be \iiisuitable. So the block $\AllOrbit{M}{\varphi} = \antichain{M}{\selection{M}{\worlds(M)}{\Cond{\varphi}}}$, which we will label $\allorblabel{1}$, can be used to \condiv-cover all orbits. Finally, to \condiv-cover $\allorblabel{1}$ we can add two more copies of $\AllOrbit{M}{\varphi}$ (labeled $\allorblabel{2}$ and $\allorblabel{3}$) and have a non-transitive loop on these three copies, which will $\condiv$-cover each other circularly. This leads to the following small model construction for \E.

\begin{definition}{\bf(Small Model Construction for \E)}\\
If $M$ is an \mbox{\E-countermodel} for $\varphi$, $\smc{\E} = \tuple{L, \preffeql, \labeling}$ where ${L = \{\fallabel{v}, \orblabel{v} \mid v \in \Fal \}} \cup \{ \allorblabel{i} \mid i \in \{1,2,3\} \}$, $\labeling(\fallabel{v}) = \antichain{M}{\{v\}}$, $\labeling(\orblabel{v}) = \Orbit{M}{\varphi}{v}$, $\labeling(\allorblabel{i}) = \AllOrbit{M}{\varphi}$ and the preference relation $\preffeql$ on blocks is demonstrated on \figureword~\ref{fig:small_model_constructions}a.
\end{definition}

\begin{theorem}
If $M$ is a \mbox{\E-countermodel} for $\varphi$ then $\generated{\smc{\E}}$ is a \mbox{\E-countermodel} for $\varphi$ and $|\worlds(\generated{\smc{\E}})| = \O(|\varphi|^2)$.\footnote{As usual, the notation $f(\varphi, M) = \O(g(\varphi, M))$ for integer-valued functions $f$ and $g$ means that there exists a constant $C$ such that $f(\varphi, M) \le C \cdot g(\varphi, M)$ for all $\varphi$ and $M$.}
\end{theorem}
\begin{proof}
$\smc{\E}$ is a countermodel for $\varphi$ by \theoremword~\ref{lem:composite_countermodel_conditions}, because $\labeling(\orblabel{v})$ is \iiisuitable for $\labeling(\fallabel{v})$ and all blocks are $\condiv$-covered by \lemmaword~\ref{lem:cond_iv_for_selection}.
$|\worlds(\generated{\smc{\E}})| =
\O(|\varphi|^2)$ since $\smc{\E}$ contains ${(2 \cdot |\Fal| + 3)
}$ blocks with at most $|\Cond{\varphi}|$ worlds each.
\end{proof}

\subsection{Small Model Construction for Logic \F}

For logic \F, we will utilize the limitedness of the countermodel $M$ to construct a small countermodel with an acyclic $\preff$, which will automatically make it limited (and thus an \F-countermodel) too. 

\pagebreak

\begin{lemma}
\label{lem:acyclic_limited}
Model $\tuple{W, \preffeq, \Val}$ is limited if $W$ is finite and $\preff$ is acyclic.
\end{lemma}
\begin{proof}
If there is some $w_0 \in \tset{\alpha}{M}$, consider (any) longest path ${ w_0 \prec w_1 \prec w_2 \prec \dots }$ with worlds from $\tset{\alpha}{M}$ staring from $w_0$. Since $W$ is finite and there can be no repetitions on the path (due to acyclicity of $\preff$), the path is finite and there is the last world $w_m$ for which there is no $u \in \tset{\alpha}{M}$ such that $u \preff w_m$, and so $w_m \in \maximalSet{\preff}{\tset{\alpha}{M}}$ by definition.
\end{proof}

For acyclicity, we will modify our construction $\smc{\E}$ by replacing a non-transitive cycle on blocks $\allorblabel{1},\allorblabel{2},\allorblabel{3}$ with one finite chain.
Our goal is to select a chain of polynomial size that is \condiv-safe and satisfies any $\alpha \in \Cond{\varphi}$ that is satisfiable in $M$ (which will allow us to use the chain to \condiv-cover any block). We construct such a chain through an iterative process,
that selects maximal worlds for disjunctions of conditions.
At the beginning of the process, we have $\CondSet_0 = \Cond{\alpha}$ as the set of conditions for which we need satisfying worlds. If at least one of conditions in $\CondSet_0$ is satisfied by some world in $M$, then $\tset{\bigvee_{\alpha \in \CondSet_0} \alpha}{M} \neq \emptyset$, then by limitedness there exists some ${z_0 \in \maximalSet{\preff}{\tset{\bigvee_{\alpha \in \CondSet_0} \alpha}{M}}}$. We can safely take $z_0$ as the first (i.e. most preferable) world in the chain, since there are no worlds $u \preff z_0$ in $M$ satisfying conditions from $\CondSet_0$. $z_0$ satisfies some conditions from $\CondSet_0$ (since $M, z_0 \models \bigvee_{\alpha \in \CondSet_0} \alpha$), therefore we can move on to the next step with a strictly smaller set $\CondSet_1$ of conditions for which we still need satisfying worlds. We can safely repeat this process by taking worlds from $z_i \in \maximalSet{\preff}{\tset{\bigvee_{\alpha \in \CondSet_i} \alpha}{M}}$ at every iteration: $z_i$ has no $\preff$-prefferable $\alpha$-worlds for all remaining conditions $\alpha \in \CondSet_i$, while for all already removed formulas there is a satisfying world somewhere earlier (i.e. preferrable to $z_i$) in the chain, thus condition~\condiv will be satisfied for this world. 
After a linear number of iterations, the chain will contain satisfying worlds for all formulas from $\CondSet$ satisfiable in $M$.

Below is the formal definition of the described chain of maximal worlds. We give a generalized version that selects maximal worlds from any given subset of worlds $U$ and any given set of formulas $\CondSet_i$, the same way as we did for $\selection{M}{U}{\CondSet}$. We will need this generalized version for logics \FCM and \G.  To define linear order in chains formally we will use the notation of lists: $\emptyseq$ will denote an empty list, and $\apseq{a}{S}$ will denote the list in which element $a$ is appended to the beginning of the list $S$.

\begin{definition}
For any $U \subseteq \worlds(M)$ and a finite set of formulas $\CondSet_i$,
\[ \maxseq{M}{U}{\CondSet_i} =
\begin{cases}
\emptyseq, & \textit{if }\CondSet_i = \emptyset \\
\emptyseq, & \textit{if }\disjset{M}{U}{\CondSet_i} = \emptyset \\
\apseq{\disjbest{M}{U}{\CondSet_i}}{\maxseq{M}{U}{\CondSet_{i+1}}}, & \textit{otherwise}
\end{cases} \]
\noindent where $\disjset{M}{U}{\CondSet} = U \cap \maximalSet{\preff}{\tset{\bigvee_{\alpha \in \CondSet_i} \alpha}{M}}$, $\disjbest{M}{U}{\CondSet} = \rep(\disjset{M}{U}{\CondSet})$ and ${\CondSet_{i+1} = \{ \alpha \in \CondSet_i \mid M, \disjbest{M}{U}{\CondSet} \not\models \alpha \}}$.
\end{definition}

Notice that for a finite $\CondSet_0$ this sequence is well-defined (representative $\disjbest{M}{U}{\CondSet_i}$ is always taken from a non-empty set and $|\CondSet_i|$ decreases) and always has length at most $|\CondSet_0|$. The reasoning above that shows \condiv-safeness of the chain built from this sequence works in the general case with arbitrary $U$ and does not even require the limitedness of $M$.

\begin{lemma}
\label{lem:chain_iv_safe}
$\chain{M}{\maxseq{M}{U}{\Cond{\varphi}}}$ is \condiv-safe for any $U \subseteq \worlds(M)$.
\end{lemma}
\begin{proof}
Let $z_k \in \worlds(\chain{M}{\maxseq{M}{U}{\Cond{\varphi}}})$ and $\Ob{\gamma}{\alpha} \in \PosOb$. $z_k \in \maximalSet{\preff}{\tset{\bigvee_{\alpha \in \CondSet_k} \alpha}{M}}$ for some step $k$ and set $\CondSet_k$ of remaining conditions. If there is $u \preff z_k$ such that $M, u \models \alpha$,
then $\alpha \not\in \CondSet_k$ due to maximality of $z_k$, which means that $\alpha$ was removed at some previous step, therefore there is $z_j$ with $j < k$ such that $M, z_j \models \alpha$.
\end{proof}

For logic \F, we select worlds in the chain from the whole $\worlds(M)$: for a limited model $M$ we define block $\MaxChain{M}{\varphi} = \chain{M}{\maxseq{M}{\Cond{\varphi}}{\worlds(M)}}$, which we will label $\maxchainlabel$. $\MaxChain{M}{\varphi}$ contains satisfying worlds for all conditions from $\Cond{\varphi}$ satisfiable in $M$ so it \condiv-covers any block on $M$.

\begin{lemma}
\label{lem:chain_iv_covers_F}
For a limited $M$, $MaxChain(M, \varphi)$ \condiv-covers any block. 
\end{lemma}
\begin{proof}
If a condition $\alpha \in \Cond{\varphi}$ is satisfiable in $M$ then it can not be among the remaining conditions when the chain is built (otherwise $\CondSet_m \neq \emptyset$ and $\disjset{M}{\worlds(M)}{\CondSet_m} \neq \emptyset$ due to limitedness of $M$), therefore for some world $z$ in the chain $M, z \models \alpha$.
\end{proof}

Replacement of non-transitive triangle in $\smc{E}$ with $\MaxChain{M}{\varphi}$ gives us the small model construction $\smc{F}$ with an acyclic strict version of preference relation.

\begin{definition}{\bf(Small Model Construction for \F)} If $M$ is an \mbox{\F-countermodel} for $\varphi$, $\smc{\F} = \tuple{L, \preffeql, \labeling}$ where ${L = \{\fallabel{v}, \orblabel{v} \mid v \in \Fal \}} \cup \{ \maxchainlabel \}$, $\labeling(\fallabel{v}) = \antichain{M}{\{v\}}$, $\labeling(\orblabel{v}) = \Orbit{M}{\varphi}{\Bet{\preff}{v}}$, $\labeling(\maxchainlabel) = \MaxChain{M}{\varphi}$ and a preference relation $\preffeql$ on blocks is demonstrated on \figureword~\ref{fig:small_model_constructions}b.
\end{definition}

\begin{theorem}
If $M$ is a \mbox{\F-countermodel} for $\varphi$ then $\generated{\smc{\F}}$ is a \mbox{\F-countermodel} for $\varphi$ and $|\worlds(\generated{\smc{\F}})| = \O(|\varphi|^2)$.
\end{theorem}
\begin{proof}
$\generated{\smc{\F}}$ is a countermodel for $\varphi$ by \theoremword~\ref{lem:composite_countermodel_conditions}, because $\labeling(\orblabel{v})$ is \iiisuitable for $\labeling(\fallabel{v})$, $\labeling(\maxchainlabel)$ is $\condiv$-safe by \lemmaword~\ref{lem:chain_iv_safe} and all other blocks are $\condiv$-covered by \lemmaword~\ref{lem:chain_iv_covers_F} and \lemmaword~\ref{lem:cond_iv_for_selection}.
$\generated{\smc{\F}}$ is an $\F$-countermodel by \lemmaword~\ref{lem:acyclic_limited}.
$|\worlds(\generated{\smc{\F}})| =
\O(|\varphi|^2)$ since $\smc{\F}$ contains ${(2 \cdot |\Fal| + 1)}$ blocks with at most $|\Cond{\varphi}|$ worlds each.
\end{proof}

\subsection{Small Model Construction for Logic \FCM}

For logic \FCM we need to ensure the transitivity of the preference relation in $\generated{\smc{\FCM}}$. It is enough to obtain an \FCM-countermodel since for finite models transitivity implies smoothness.

\pagebreak

\begin{lemma}
\label{lem:transitive-smooth}
$M = \tuple{W, \preffeq, \Val}$ is smooth if $W$ is finite and $\preffeq$ is transitive.
\end{lemma}
\begin{proof}
First, notice that transitivity of $\preffeq$ implies transitivity of $\preff$. Indeed, if $w_1 \preff w_2$ and $w_2 \preff w_3$ then $w_1 \preffeq w_3$ by transitivity of $\preffeq$ and $w_3 \not\preffeq w_1$ since otherwise there would be a transitive triangle on these three worlds and none of them could be strictly preferable to another. Now, for an arbitrary $w_0 \in \tset{\alpha}{M}$, consider (any) longest path ${ w_0 \prec w_1 \prec w_2 \prec \dots }$ with worlds from $\tset{\alpha}{M}$ staring from $w_0$. Since $W$ is finite and there can be no repetitions on the path due to transitivity of $\preff$, the path is finite and there is the last world $w_m$ for which there is no $u \in \tset{\alpha}{M}$ such that $u \preff w_m$, so $w_m \in \maximalSet{\preff}{\tset{\alpha}{M}}$ and either $w_m = w_0$ or $w_m \preff w_0$, so $M$ is smooth.
\end{proof}

In $\smc{\F}$ non-transitivity was essential: we can not put $\maxchainlabel \preffl \fallabel{v}$ since we selected worlds in the maximal chain from the whole initial model, so $\labeling(\maxchainlabel)$  can be not \iiisuitable for $\labeling(\fallabel{v})$. However, smoothness allows us to select a maximal chain only among worlds in $\Bet{\preff}{v}$.
Specifiaclly, for every falsifying world $v$ we introduce individual \emph{chain-orbit} $\MaxOrbChain{M}{\varphi}{v} = \chain{M}{\maxseq{M}{\Bet{\preff}{v}}{\Cond{\varphi}}}$. We already know that this block is \mbox{\condiv-safe} by \lemmaword~\ref{lem:chain_iv_safe}, and we can show that for an \mbox{\FCM-model} $M$ it covers block $\labeling(\fallabel{v})$.

\begin{lemma}
\label{lem:chain_iv_covers_FCM}
For a transitive and smooth $M$, $\MaxOrbChain{M}{\varphi}{v}$ \condiv-covers $\antichain{M}{\{v\}}$.
\end{lemma}
\begin{proof}
Suppose that (1) there is some $u \preff v$ in $M$ such that $M, u \models \alpha$ for some $\alpha \in \Cond{\varphi}$, we need to show that there is a world $z_k$ in the chain such that $M, z_k \models \alpha$. Similarly to \lemmaword~\ref{lem:chain_iv_covers_F}, we show it by proving that in this case $\alpha$ is removed from the set of conditions $\CondSet_i$ at some point. And to show this, it is enough to prove that for smooth models (*) $\alpha \in \CondSet_i$ implies $\disjset{M}{\Bet{\preff}{v}}{\CondSet_i} \neq \emptyset$ from the definition of $\maxseqword$ (then the sequence of maximal worlds cannot end while $\alpha$ belongs to $\CondSet_i$).

Let us prove (*). Suppose that $\alpha \in \CondSet_i$. From this and (1) follows $u \in \tset{\bigvee_{\alpha \in \CondSet_i} \alpha}{M}$. Due to smoothness of $M$ it implies that (2) there is ${u' \in \maximalSet{\preff}{\tset{\bigvee_{\alpha \in \CondSet_i} \alpha}{M}}}$ such that either $u' = u$ or $u' \preff u$. In either case $u' \preff v$ (since $u \preff v$ by (1) and transitivity of $\preffeq$ implies transitivity of $\preff$). So, we have (3) $u' \in \Bet{\preff}{v}$. (2) and (3) together imply $u' \in \disjset{M}{\Bet{\preff}{v}}{\CondSet_i}$, concluding the proof of (*).
\end{proof}

So we can obtain a small model construction for \FCM by replacing each orbit $\orblabel{v}$ with an individual maximal chain $\MaxOrbChain{M}{\varphi}{v}$ (which we will label $\betmaxchainlabel{v}$). The common chain $\maxchainlabel$ from $\smc{\F}$ is not needed anymore.

\begin{definition}{\bf(Small Model Construction for \FCM)}\\
If $M$ is an \mbox{\FCM-countermodel} for $\varphi$, $\smc{\FCM} = \tuple{L, \preffeql, \labeling}$ where $L = \{\fallabel{v}, \betmaxchainlabel{v} \mid v \in \Fal \}$, $\labeling(\fallabel{v}) = \antichain{M}{\{v\}}$, $\labeling(\betmaxchainlabel{v}) = \MaxOrbChain{M}{\varphi}{v}$ and a preference relation $\preffeql$ on blocks is demonstrated on \figureword~\ref{fig:small_model_constructions}c.
\end{definition}

\begin{theorem}
If $M$ is a \mbox{\FCM-countermodel} for $\varphi$ then $\generated{\smc{\FCM}}$ is a \mbox{\FCM-countermodel} for $\varphi$ and $|\worlds(\generated{\smc{\FCM}})| = \O(|\varphi|^2)$.
\end{theorem}
\begin{proof}
$\smc{\FCM}$ is a countermodel for $\varphi$ by \theoremword~\ref{lem:composite_countermodel_conditions}, because $\labeling(\betmaxchainlabel{v})$ is \iiisuitable for $\labeling(\fallabel{v})$, each $\labeling(\betmaxchainlabel{v})$ is $\condiv$-safe by \lemmaword~\ref{lem:chain_iv_safe} and each $\labeling(\fallabel{v})$ is $\condiv$-covered by $\labeling(\betmaxchainlabel{v})$ by \lemmaword~\ref{lem:chain_iv_covers_FCM}.
$\smc{\FCM}$ is an $\FCM$-countermodel by \lemmaword~\ref{lem:transitive-smooth}.
$|\worlds(\generated{\smc{\FCM}})| =
\O(|\varphi|^2)$ since $\smc{\FCM}$ contains ${(2 \cdot |\Fal|)}$ blocks with at most $|\Cond{\varphi}|$ worlds each.
\end{proof}

\begin{remark}
\label{rem:Friedman_Halpern_FCM}
The form of the countermodel that we obtain --- a union of incomparable finite chains --- is the same as a Friedman-Halpern countermodel for logic \PCA (i.e. \FCM)~\cite{Friedman_Halpern_small_model_constructions}. However, we have achieved it by using different methods: they use a finite-model property of \PCL extensions (shown in~\cite{Burgess}) and extend the preference relation to a linear order, then construct chains by selecting the greatest world w.r.t. extended order independently for each conditional in $\PosOb$, while we do it using an iterative procedure.
The possibility of their selection fully relies on finitedness and transitivity, which due to \lemmaword~\ref{lem:transitive-smooth} is only possible in smooth models, so it cannot be applied to the weaker logics \E and \F. 
Furthermore, the Horn fragment\footnote{Conditional Horn formula is a formula of a form $\Ob{\gamma_1}{\alpha_1} \wedge \dots \wedge \Ob{\gamma_n}{\alpha_n} \to \Ob{\gamma_0}{\alpha_0}$.} of \PCA was studied extensively in the area of non-monotonic reasoning, where it is known as the KLM logic \KLMP~\cite{KLM} of preferential reasoning. A small model construction for \KLMP has been introduced in~\cite{KLM_entailment} and consists of a single chain of polynomial size (by essentially the same method as Friedman-Halpern). Notice, that both Friedman-Halpern and our constructions turn into a single chain when restricted to Horn formulas.
\end{remark}

\subsection{Small Model Construction for Logic \G}

For logic \G, we also need to ensure the totalness of the transformed model by leveraging the fact that the falsifying worlds in $\Fal$ are ordered in the initial model by $\preffeq$ which in a \G-model is a total preorder.

Let us consider first a simple case where $\preffeq$ in the given \G-countermodel is asymmetric (and therefore a strict linear order). Then there exists an ordering $v_1 \prec \dots \prec v_n$ of worlds from $\Fal$. Then we can linearly order blocks of $\smc{\FCM}$ with the following order: ${\fallabel{v_1} \succl \betmaxchainlabel{v_1} \succl \dots \succl \fallabel{v_n} \succl \betmaxchainlabel{v_n}}$. The \condiii-suitability will still be satisfied with such ordering, because for every $j \ge i$ we have ${\worlds(\betmaxchainlabel{v_j}) \subseteq \Bet{\preff}{v_j}}$ and $\Bet{\preff}{v_j} \subseteq \Bet{\preff}{v_i}$ due to transitivity of $\preffeq$.


In general, $\preffeq$ is not necessarily asymmetric, but we can generalize the same idea by grouping together $\preffeq$-equivalent worlds as in the following definition.

\begin{definition}{\bf(Stratification)}
For $M = \tuple{\worlds, \preffeq, \Val}$ and a finite set $U \subseteq \worlds$, a sequence $\seq{S_1, \dots, S_n}$ of non-empty subsets of $\worlds$ is called a \emph{stratification of $U$} when $U$ is the disjoint union of subsets $\{ S_i \}_{i=1}^n$ and for every $u_i \in S_i, u_j \in S_j$ we have $s_i \preffeq s_j$ iff $i \ge j$.
\end{definition}

For total preorders the unique stratification of any finite set is given by its factorization w.r.t. $\preffeq$-equivalence.

\begin{lemma}
If $\preffeq$ is transitive and total, there exists a unique stratification of any finite subset $U$.
\end{lemma}
\begin{proof}
Consider an equivalence relation $\preffequiv$ on $U$ where $u_1 \preffequiv u_2$ means that both $u_1 \preffeq u_2$ and $u_2 \preffeq u_1$. Consider further a relation $\preffsort$ on the set of equivalence classes of $U$ w.r.t. $\preffequiv$ where $S_i \preffsort S_j$ when there exist $u_i \in S_i$ and $u_j \in S_j$ such that $u_i \preffeq u_j$. Notice that for a transitive and total $\preffeq$ the relation $\preffsort$ is a linear order: it is antisymmetric due to definitions of $\preffequiv$ and $\preffsort$, and it is transitive and total (and hense reflexive) due to the transitivity and totalness of $\preffeq$. This linear ordering gives a stratification by definition. Notice also that it is the only stratification: every element of a stratification should be an equivalence class w.r.t. $\preffequiv$ and their order in the list should be aligned with $\preffsort$ (i.e. $i \ge j$ implies $S_i \preffsort S_j$) by definition.
\end{proof}

Therefore, we can take the stratification $\seq{S_1, \dots, S_n}$ of $\Fal$ w.r.t. $\preffeq$ and create a block $\clique{M}{S_i}$ for every group $S_i$.

Notice that for $u_1, u_2 \in S_i$ both $u_1 \preffeq u_2$ and $u_2 \preffeq u_1$ so $\Bet{\preff}{u_1} = \Bet{\preff}{u_2}$ (due to transitivity). Therefore, to \condiv-cover block $\clique{M}{S_i}$ we can take orbit-chain $\MaxOrbChain{M}{\varphi}{\rep(S_i)}$ with arbitrary representative of $S_i$.

\begin{lemma}
\label{lem:chain_iv_covers_G}
For a transitive and smooth $M = \tuple{W, \preffeq, \Val}$ and $S \subseteq \worlds(M)$, if $u \preffeq u'$ for all $u, u' \in S$ then $\MaxOrbChain{M}{\varphi}{\rep(S)}$ \condiv-covers $\clique{M}{S}$.
\end{lemma}
\begin{proof}
For any $u \in S_i$ holds $\Bet{\preff}{u} = \Bet{\preff}{\rep(S)}$ (since $\preffeq$ is transitive), so any formula from $\Cond{\varphi}$ satisfied in some world from $\Bet{\preff}{v}$ is also satisfied by some world in $\MaxOrbChain{M}{\varphi}{\rep(V)}$ by \lemmaword~\ref{lem:chain_iv_covers_FCM}.
\end{proof}

Thus we can take as the construction $\smc{\G}$ a linearly-oredered sequence of blocks in which cliques $\clique{M}{S_i}$, labeled $\grfallabel{S_i}$, are interleaved with chain-orbits $\MaxOrbChain{M}{\varphi}{\rep(S_i)}$, labeled $\grbetmaxchainlabel{S_i}$.

\begin{definition}{\bf(Small Model Construction for \G)}
For a \mbox{\G-countermodel} $M = \tuple{\worlds, \preffeq, \Val}$ for a formula $\varphi$,
let $\seq{S_1, \dots, S_n}$ be the unique stratification of $\Fal$ w.r.t. $\preffeq$. Then $\smc{\G} = \tuple{L, \preffeql, \labeling}$ where $L = \bigcup_{i=1}^n \{\grfallabel{S_i}, \grbetmaxchainlabel{S_i} \}$, $\labeling(\grfallabel{S_i}) = \clique{M}{V_i}$, $\labeling(\grbetmaxchainlabel{S_i}) = \chain{M}{\MaxOrbChain{M}{\varphi}{\rep(S_i)}}$ and the blocks are ordered linearly as follows: ${\grfallabel{S_1} \succl \grbetmaxchainlabel{S_1} \succl \dots \succl \grfallabel{S_n} \succl \grbetmaxchainlabel{S_n}}$.
\end{definition}

\begin{theorem}
If $M$ is an \mbox{\G-countermodel} for $\varphi$ then $\generated{\smc{\G}}$ is a \mbox{\G-countermodel} for $\varphi$ and $|\worlds(\generated{\smc{\G}})| = \O(|\varphi|^2)$.
\end{theorem}
\begin{proof}
$\smc{\G}$ is a countermodel for $\varphi$ by \theoremword~\ref{lem:composite_countermodel_conditions}: $\labeling(\grfallabel{S_i})$ is flat and all block $\preffl$-preferable to it are \iiisuitable for it (since $\preffeq$ in $M$ is transitive), each $\labeling(\grbetmaxchainlabel{S_i})$ is $\condiv$-safe by \lemmaword~\ref{lem:chain_iv_safe} and each $\labeling(\grfallabel{S_i})$ is $\condiv$-covered by $\labeling(\grbetmaxchainlabel{S_i})$ due to \lemmaword~\ref{lem:chain_iv_covers_G}.
$\smc{\G}$ is a \mbox{$\G$-countermodel} since its preference relation is transitive and total, and also smooth by \lemmaword~\ref{lem:transitive-smooth}.
$|\worlds(\generated{\smc{\G}})| =
\O(|\varphi|^2)$ since $\smc{\G}$ contains at most ${(2 \cdot |\Fal|)}$ blocks with at most $|\Cond{\varphi}|$ worlds each.
\end{proof}

\begin{remark} Friedman and Halpern also provide a counter-model for logic \VTA (i.e. \G)~\cite{Friedman_Halpern_small_model_constructions}. They use an ad hoc approach, different from the one they use for the other extensions of \PCL. For \VTA for each conditional they simply take one world from the original model without
changing the preference relation, resulting in a model of linear size. 
Although \theoremword~\ref{lem:rearranged_countermodel_conditions} can be also used to establish the adequacy of their constriction,
we provided a different construction (with a new explicitly defined preference relation and with potentially a quadratic number of worlds) for uniformity with the constructions for the other three \AA qvist's logics. 
\end{remark}

\section{Applications}

In this section, we describe two applications of our small model constructions: alternative semantical characterizations, complexity and encodings in the classical propositional logic.

\subsection{Alternative semantical characterizations}
\label{sec:alternative_semantics}

We will call a class $\modelclass$ of preference models a semantical characterization for theoremhood in logic $\logicsymb$ when any $\varphi$ is a theorem of $\logicsymb$ iff $\varphi$ is valid in all models from $\modelclass$.
New semantical characterizations for theoremhood can be extracted from the specific form of our small model constructions. Namely, any model property satisfied by $\smc{\logicsymb}$ that is stronger than some existing characterization for $\logicsymb$ (e.g. from \figureword~\ref{tab:initial_taxonomy}) can be used as an alternative characterization.

\begin{lemma}
\label{lem:alternative_semantics_generator}
Let $\modelclass$ be a class of models characterizing theoremhood in $\logicsymb$. If $\modelclass' \subseteq \modelclass$ and $\smc{\logicsymb} \in \modelclass'$ for every $\varphi$ and $M$, then $\modelclass'$ also characterizes theoremhood in $\logicsymb$.
\end{lemma}
\begin{proof}
If $\varphi$ is a theorem of $\logicsymb$, it is valid in all models in $\modelclass$, so it is also valid in all models of $\modelclass'$. If $\varphi$ is not a theorem of $\logicsymb$, it is not valid in $\smc{\logicsymb}$ that belongs to $\modelclass'$.
\end{proof}




We can use this method to characterize theoremhood in \AA qvist logics with \emph{frame properties}, i.e. properties of the preference relation. Notice that the limit conditions (limitedness and smoothness) used for the characterization of \F, \FCM, and \G are not frame properties: they impose conditions only on truth sets of the model. This choice plays a vital role in establishing correspondence between semantics and known axiomatizations of \AA qvist's logics, but it makes it hard to work with these models since you need to distinguish which subsets of worlds can be a truth set. However, we can notice that our small model constructions satisfy the corresponding limit conditions for all subsets of worlds, therefore limitedness/smoothness on the level of frames can be used to characterize theoremhood also.

More importantly, our constructions satisfy some stronger frame properties. We already used these properties to prove that $\smc{\logicsymb}$ generates an \mbox{$\logicsymb$-countermodel}.
If we consider only finite models limit conditions can be replaced with natural conditions on preference relations: limitedness can be replaced with acyclicity, and smoothness can be dropped in presence of transitivity.



\begin{figure}
    \centering

    \begin{tabular}{|c|c|P{15mm}|P{15mm}|P{15mm}|}
    \hline
    \multirow{2}{*}{Logic} & \multirow{2}{*}{Cardinality of $W$} & \multicolumn{3}{c|}{Properties of $\preffeq$} \\
    \hhline{|~|~|-|-|-|}
     & & acyclic & transitive & total  \\
    \hline
    $\E$ & finite & & & \\
    \hline
    $\F$ & finite & \yes & & \\
    \hline
    $\FCM$ & finite & & \yes & \\
    \hline
    $\G$ & finite & & \yes & \yes \\
    \hline
    \end{tabular}
    \caption{Finite-model characterizations of theoremhood in \AA qvist's logics (with maximality as the notion of bestness). }
    \label{tab:finite_model_taxonomy}
\end{figure}

\begin{theorem}
\label{lem:finite_model_semantics}
Formula $\varphi$ is a theorem of \AA qvist logic $\logicsymb$ iff $M \models \varphi$ for all finite models $M$ satisfying the frame properties for logic $\logicsymb$ from \figureword~\ref{tab:finite_model_taxonomy}.
\end{theorem}
\begin{proof}
Using \lemmaword~\ref{lem:alternative_semantics_generator}. Finite models with transitive preference realtion are smooth by \lemmaword~\ref{lem:transitive-smooth}. Finite models with acyclic preference relation are limited by \lemmaword~\ref{lem:acyclic_limited} (acyclicity of $\preffeq$ implies acyclicity of $\preff$), and although the preference relation in $\smc{\F}$ has cycles in the form of reflexive loops inside the chain, these loops can be removed without affecting the satisfaction in the model, so any non-theorem has a finite countermodel with acyclic preference relation.
\end{proof}

In addition, our models for \E, \F, \FCM satisfy antisymmetry so this property can be added to finite-model characterization for these logics from \figureword~\ref{tab:finite_model_taxonomy}, but not to the characterization of \G.\footnote{E.g., the principle of conditional excluded middle $\Ob{\gamma}{\alpha} \lor \Ob{\neg \gamma}{\alpha}$ is valid in all models with total and antisymmetric preference relation (since they have at most one $\alpha$-best world), but it can be easily falsified in \G by a model with two $\alpha$-worlds preferable to each other (and thus both allowed to be $\alpha$-best simultaniously).} Either reflexivity or irreflexivity can also be added since it is trivial to force them in any model without changing the satisfaction relation. Thus, \FCM is characterized by finite models where $\preffeq$ is a partial order (strict or non-strict). At the same time, finite models where $\preffeq$ is a linear order give some logic that is stronger than \G.
Even more specialized properties can be extracted from our construction via \lemmaword~\ref{lem:alternative_semantics_generator}, e.g. \FCM can be characterized by models that are unions of non-comparable finite chains.

\begin{remark}
Note that here we are only concerned with semantical characterizations of theoremhood. Our results can not be extended to characterizing entailments $\Gamma \vdash \varphi$ if $\Gamma$ is infinite (to provide a \emph{strongly complete} characterization of the logics). This is a natural limitation for finite-model characterizations since entailments from infinite sets of premises can not be characterized using only finite models for \Sfive already (see a counterexample in \appendixword~\ref{sec:strong_completeness}) and therefore for \AA qvist's logics that extend \Sfive too.
\end{remark}

\subsection{Complexity and automated deduction}

Our small model constructions show that for any non-valid formula there exists a countermodel with at most $N(\varphi)$ worlds, where $N(\varphi)$ is a certain upper bound polynomial w.r.t. $|\varphi|$. 
Plus, the stronger frame properties from \figureword~\ref{tab:finite_model_taxonomy} can be easily checked in polynomial time w.r.t. the model size. This immediately implies co-NP-completeness of theoremhood.


\begin{theorem}
\label{lem:finite_characterization}
Theoremhood is co-NP-complete for every \AA qvist's logic.
\end{theorem}
\begin{proof}
Non-theoremhood can be checked non-deterministically in polynomial time by guessing a countermodel $M$ of size at most $N(\varphi)$ (i.e. guessing preference relation and valuation for all variables occurring in $\varphi$) and then checking $M \not\models \varphi$ and the required properties from \figureword~\ref{tab:finite_model_taxonomy}. \mbox{co-NP-hardness} follows from co-NP-completeness of theoremhood in classical logic (since a propositional formula is a classical tautology iff it is a theorem of an \AA qvist logic).
\end{proof}

Moreover, with simpler finite-model characterization from \figureword~\ref{tab:finite_model_taxonomy} a countermodel definition can be naturally encoded with a propositional formula of a polynomial size (see \appendixword~\ref{sec:SAT_encoding} for the full encodings).
This propositional formula can be given to any SAT-solver for efficient theoremhood checking and countermodels can be reconstructed from classical models found by the solver.


\section*{Concluding remark}
In this paper, we provide small model constructions for \AA qvist's logics, which can be used to understand theoretical properties of these logics (such as finite-model semantical characterizations and complexity) and to generate countermodels for non-valid formulas using
SAT-solvers.
Ideally, this should be complemented by analytic calculi which provide transparent derivations for valid formulas. 
We plan to explore the relationship between our constructions and hypersequent calculi, aiming for simpler proof-theoretic characterizations, particularly for the challenging logic \F.

\section*{Acknowledgements}

I want to thank Agata Ciabattoni and Dominik Pichler for the helpful discussions and comments on early versions of this paper, Xavier Parent for the explanations regarding the history of preference-based approach in deontic logic, Roman Kuznets for the discussion about the relation between the strong completeness and the finite model property, and the anonymous reviewers for their useful remarks and suggestions.

\vspace{1mm}
\begin{minipage}{0.12\textwidth}
\hspace{-4mm}
\includegraphics[width=1.5cm]{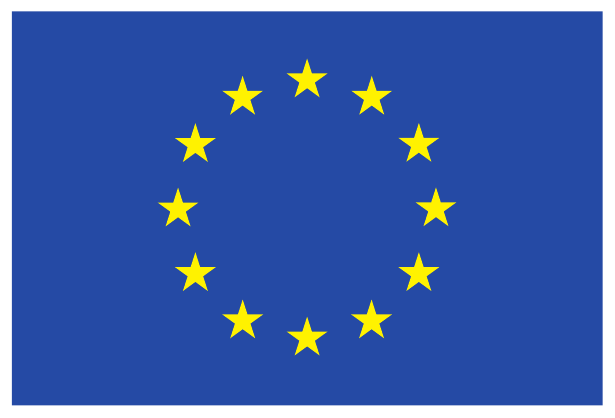}
\end{minipage}%
\hfill
\begin{minipage}{0.83\textwidth}
This work was funded by the European Union’s Horizon 2020 research and innovation programme under grant agreement No 101034440.
\end{minipage}
\bibliographystyle{aiml22}
\bibliography{main}

\newpage
\Appendix

\section{Detailed proof of \theoremword~\ref{lem:rearranged_countermodel_conditions}}
\label{sec:inductive_details}

\begin{theorem}
Let $\varphi$ be a formula and $M = \tuple{W, \preffeq, \Val} \in \Models$ such that ${M \not\models \varphi}$. If a model $M' = \tuple{W', \preffeq', \Val'} \in \Models$ rearranges $M$ with the prototype function $\prot \colon W' \to W$ then the following four conditions are sufficient for $M' \not\models \varphi$.
\begin{enumerate}
\item[\condi] There exists $v' \in W'$ such that $M, \prot(v') \not\models \varphi$.
\item[\condii] For any $\Box \beta \in \NegBox$ there exists $v' \in W'$ such that $M, \prot(v') \not\models \beta$.
\item[\condiii] For any ${\Ob{\gamma}{\alpha} \in \NegOb}$ there exists ${v' \in W'}$ such that ${\prot(v') \in \maximalSetM{\tset{\alpha}{M}} \setminus \tset{\gamma}{M}}$ and for all ${u' \preff' v}'$ holds ${\prot(u') \preff \prot(v')}$. 
\item[\condiv] For any $w' \in W'$, for all $\Ob{\gamma}{\alpha} \in \PosOb$ if there exists $u \preff \prot(w')$ such that $M, u \models \alpha$ then there exists $u' \preff' w'$ such that $M, \prot(u') \models \alpha$.
\end{enumerate}
\end{theorem}
\begin{proof}
We will prove that for any $w' \in W'$ and any $\psi \in \SubF{\varphi}$ holds $M', w' \models \psi$ iff $M, \prot(w') \models \psi$. Then $M' \not\models \varphi$ follows by the condition \condi.
The proof is by induction on $\psi$ (we use the abbreviation IH(s) to refer to the inductive hypothesis(-es)).
\begin{itemize}
\item $\psi = x \in \PropVars$. $w \in \Val(x)$ iff $\prot(w) \in \Val'(x)$ by the definition of the prototype function.
\item $\psi = \neg \psi'$. Directly from IH for $\psi'$.
\item $\psi = \psi_1 \wedge \psi_2$. Directly from IHs for $\psi_1$ and $\psi_2$.
\item $\psi = \Box \beta$ and $\psi \in \PosBox$. $M \models \Box \beta$, so for all $w' \in W'$ holds $M, \prot(w') \models \beta$, so by IH for all $w' \in W'$ holds $M', w' \models \beta$, so $M' \models \Box \beta$.
\item $\psi = \Box \beta$ and $\psi \in \NegBox$. By \condii there is $v' \in W'$ such that $M, \prot(v') \not\models \beta$, so by IH holds $M', v' \not\models \beta$, so $M' \not\models \Box \beta$.
\item $\psi = \Ob{\gamma}{\alpha}$ and $\psi \in \PosOb$. Take any $w' \in W'$ such that $w' \in \maximalSet{\preff'}{\tset{\alpha}{M'}}$. Then (1) $\prot(w') \in \maximalSet{\preff}{\tset{\alpha}{M}}$: $\prot(w') \in \tset{\alpha}{M}$ by IH, and there can be no $s \preff \prot(w')$ such that $s \in \tset{\alpha}{M}$ (otherwise there would be $u' \preff' w'$ such that $M', u' \models \alpha$ by \condiv and IH). Since $M \models \Ob{\gamma}{\alpha}$, (1) implies $M, \prot(w') \models \gamma$, which implies $M', w' \models \gamma$ by IH. Thus $M' \models \Ob{\gamma}{\alpha}$.
\item $\psi = \Ob{\gamma}{\alpha}$ and $\psi \in \NegOb$. For the corresponding world $v' \in W'$ from \condiii we have $v' \notin \tset{\gamma}{M'}$ (by IH and the choice of $v'$) and $v' \in \maximalSet{\preff'}{\tset{\alpha}{M'}}$ (since $v' \in \tset{\alpha}{M'}$ by IH and the choice of $v'$, and for all $u' \preff' v'$ we have $u' \notin \tset{\alpha}{M'}$ by \condiii and IH), so $M' \not\models \Ob{\gamma}{\alpha}$. 
\end{itemize}
\end{proof}

\section{Strong completeness vs finite-modal characterization}
\label{sec:strong_completeness}

In this appendix, we show that the entailment from the infinite set of premises in \Sfive can not be characterized with a class of models that contains only finite models. Consider the infinite set of propositional variables $\{ x_i \}_{i=1}^\infty$ and an infinite sequence of formulas $\{ \varepsilon_n \}_{n=1}^\infty$ defined as
\[ \varepsilon_n = x_n \wedge \bigwedge_{i=1}^{n - 1} \neg x_i \]
Clearly, two formulas $\varepsilon_n$ and $\varepsilon_m$ for $n \neq m$ can not be both satisfied in one world. Therefore, the formulas from the set $\Gamma^\Diamond = \{ \Diamond \varepsilon_n \}_{n=1}^\infty$ can not all be simultaneously valid in any finite model. So entailment $\Gamma^\Diamond \vDash \bot$ holds in all finite Kripke models, but does not hold in $\Sfive$, since an infinite Kripke model satisfying all formulas from $\Gamma^\Diamond$ simultaneously can be easily constructed.

\section{Propositional encoding for \AA qvist's logics}
\label{sec:SAT_encoding}

This appendix provides an embedding of every Åqvist's logic into the classical propositional logic. Specifically, for any given modal formula $\varphi$ in some Åqvist's logic $\logicsymb$ we define a propositional formula $\encodedf{\logicsymb}{\varphi}$, such that there is a one-to-one correspondence between classical countermodels for $\encodedf{\logicsymb}{\varphi}$ and preference countermodels for $\varphi$ with $N(\varphi)$ worlds (where $N(\varphi)$ is a size bound given by our small model construction) satisfying model conditions for $\logicsymb$ in \figureword~\ref{tab:finite_model_taxonomy}. As a result, $\varphi$ is valid in $\logicsymb$ iff $\encodedf{\logicsymb}{\varphi}$ is classically valid.

To encode a countermodel $M$ for a formula $\varphi$ with words $\{w_1, \dots, w_{N(\varphi)}\}$ we will use the following variables:
\begin{itemize}
\item $p_{i,j}$ for $1 \le i,j \le N(\varphi)$ to encode the fact $w_i \preffeq w_j$
\item $v_i^{\psi}$ for $1 \le i \le N(\varphi)$ and $\psi \in \SubF{\varphi}$ to encode the fact $M, w_i \models \psi$
\end{itemize}

If $\psi$ is not a propositional variable, $v_i^{\psi}$ is determined by {$v$-variables} for the immediate subformulas of $\psi$ and this can be straightforwardly encoded by definition by a set of propositional equivalences of polynomial size:

\[ \begin{array}{lccccccc}
\clausesmod & = & & \{ &
 v_i^{\neg \psi} \Leftrightarrow (\neg v_i^{\psi}) & \;\mid\; &
(\neg \psi) \in \SubF{\varphi} & \} _{1 \le i \le N(\varphi)} \\    
& & \cup & \{ & 
v_i^{\psi_1 \wedge \psi_2} \Leftrightarrow (v_i^{\psi_1} \wedge v_i^{\psi_2}) & \;\mid\; &
(\psi_1 \wedge \psi_2) \in \SubF{\varphi} & \} _{1 \le i \le N(\varphi) } \\
& & \cup & \{ & 
v_i^{\Box \beta} \Leftrightarrow ( \bigwedge_{j = 1}^{N(\varphi)} v_j^{\beta} ) & \;\mid\; &
(\Box \beta) \in \SubF{\varphi} & \} _{1 \le i \le N(\varphi) } \\
& & \cup & \multicolumn{5}{l}{
\{ \; v_i^{\Ob{\gamma\;}{\;\alpha}} \Leftrightarrow ( \bigwedge_{j = 1}^{N(\varphi)} (v_j^{\gamma} \vee \neg v_j^{\alpha} \vee (\bigvee_{t=1}^{N(\varphi)} (p_{t,j} \wedge \neg p_{j, t} \wedge v_t^{\alpha}) ) ) 
} \\
& & & & 
 & \;\mid\; &
\Ob{\gamma}{\alpha} \in \SubF{\varphi} & \} _{1 \le i \le N(\varphi) } \\
\end{array} \]

Transitivity and totality of $\preffeq$ can be encoded straightforwardly by definition too:

\[ \begin{array}{lccccl}
\clausestr & = & & \{ &
 (p_{i,j} \wedge p_{j,k}) \Rightarrow p_{i,k} &
\} _{1 \le i,j,k \le N(\varphi) } \\    
\clausestot & = & & \{ & 
p_{i,j} \lor p_{j,i} & 
\} _{1 \le i,j \le N(\varphi) } \\
\end{array} \]

To encode acyclicity of $\preffeq$ we can reformulate it equivalently as follows: there exists a relation $\preffeq^t$ that is transitive, irreflexive and contains $\preffeq$ (such relation exists iff the positive transitive closure of $\preffeq$ is irreflexive, i.e. when $\preffeq$ is acyclic). Introducing additional variables $t_{i,j}$ for $1 \le i,j \le N$ to encode the fact $w_i \preffeq^t w_j$, we can then encode acyclicity with the following set of formulas:

\[ \begin{array}{lccccl}
\clausesacyc & = & & \{ &
 (t_{i,j} \wedge t_{j,k}) \Rightarrow t_{i,k} &
\} _{1 \le i, j, k \le N(\varphi) } \\    
& & \cup & \{ & 
\neg t_{i,i} &
\} _{1 \le i \le N(\varphi) } \\
& & \cup & \{ & 
p_{i,j} \Rightarrow t_{i,j} & \} _{1 \le i, j \le N(\varphi) } \\
\end{array} \]

Putting everything together, we get the following encodings of countermodels (falsifying $\varphi$ in the world $w_1$) in \AA qvist's logics as formulas of polynomial size.

\[ 
\begin{array}{lcl}
\encodedf{\E}{\varphi} & = & \neg v_1^{\varphi} \wedge \bigwedge \clausesmod  \\
\encodedf{\F}{\varphi} & = & \neg v_1^{\varphi} \wedge \bigwedge \clausesmod \wedge \bigwedge \clausesacyc  \\
\encodedf{\FCM}{\varphi} & = & \neg v_1^{\varphi} \wedge \bigwedge \clausesmod \wedge \bigwedge \clausestr  \\
\encodedf{\G}{\varphi} & = & \neg v_1^{\varphi} \wedge \bigwedge \clausesmod \wedge \bigwedge \clausestr \wedge \bigwedge \clausestot \\
\end{array}
\]


\end{document}